\documentclass[a4paper, 11pt]{article}

\usepackage{amsmath,amsfonts,amssymb,amsthm,graphicx,graphics,epsf}
\usepackage{graphicx}

\usepackage{algorithmic}
\usepackage{algorithm}
\graphicspath{{img/}} %

\usepackage[margin=.9in]{geometry}

\newtheorem{theorem}{Theorem}[section]
\newtheorem{corollary}[theorem]{Corollary}
\newtheorem{lemma}[theorem]{Lemma}

\newtheorem{observation}[theorem]{Observation}

\newcommand{\conv}{\ensuremath{\Gamma}}
\newcommand{\hedge}{half-edge}
\newcommand{\hedges}{half-edges}
\newcommand{\smatching}{switch-matching}
\newcommand{\smatchings}{switch-matchings}
\newcommand{\svertex}{switch-vertex}
\newcommand{\svertices}{switch-vertices}
\newcommand{\invariant}{color-invariant}
\newcommand{\wc}{well-colored}
\newcommand{\cvisible}{color-visible}

\DeclareMathOperator{\poly}{poly}
\DeclareMathOperator{\fint}{int}

\begin{document}

\title{Linear transformation distance for bichromatic
  matchings\footnote{Research of Oswin Aichholzer supported by the ESF EUROCORES programme EuroGIGA -- CRP `ComPoSe', Austrian Science Fund (FWF): I648-N18.
Research of Thomas Hackl supported by the Austrian Science Fund (FWF): P23629-N18 `Combinatorial Problems on Geometric Graphs'.
Alexander Pilz is a recipient of a DOC-fellowship of the Austrian Academy of Sciences at the Institute for Software Technology, Graz University of Technology, Austria.}}

\author{Oswin Aichholzer\thanks{Institute for Software Technology, Graz University of Technology, Graz, Austria, {\tt [oaich|thackl|apilz|bvogt]@ist.tugraz.at}} \and
 Luis Barba\thanks{D\'epartement d'Informatique, Universit\'e Libre de Bruxelles, Brussels, Belgium, {\tt lbarbafl@ulb.ac.be}}\ \thanks{School of Computer Science, Carleton University, Ottawa, Canada} \and Thomas Hackl$^\dag$ \and Alexander Pilz$^\dag$ \and Birgit Vogtenhuber$^\dag$}

\date{}

\maketitle
\thispagestyle{empty}
\begin{abstract}
Let $P=B\cup R$ be a set of $2n$ points in general position, where $B$ is a set of $n$ blue points and $R$ a set of $n$ red points.
A \emph{$BR$-matching} is a plane geometric perfect matching on $P$ such that each edge has one red endpoint and one blue endpoint.
Two $BR$-matchings are compatible if their union is also plane. 

The \emph{transformation graph of $BR$-matchings} contains one node for each $BR$-matching and an edge joining two such nodes if and only if the corresponding two $BR$-matchings are compatible.
In SoCG 2013 it has been shown by Aloupis, Barba, Langerman, and Souvaine that this transformation graph is always connected, but its diameter remained an open question.
In this paper we provide an alternative proof for the connectivity of the transformation graph and prove an upper bound of $2n$ for its diameter, which is asymptotically tight.
\end{abstract}

\section{Introduction}

A \emph{geometric graph} $G(S,E)$ on a point set $S$ in the plane is an embedding of a graph with the point set $S$ as its vertex set and all edges embedded as straight line segments.
$G(S,E)$ is called \emph{plane} (or \emph{crossing-free}) if no two of its edges %
share a point except for a possible common endpoint.
A plane geometric graph is also called ``planar straight-line graph'' (PSLG for short).
Two plane geometric graphs $G_1(S,E_1)$ and $G_2(S,E_2)$ on the same point set are called \emph{compatible} if the union of their edge sets gives a plane geometric graph $G(S,E_1\cup E_2)$, and \emph{disjoint} if $E_1\cap E_2$ is empty.
Let $P$ be a set of $2n$ points in the plane such that $P$ does not contain three points on a common line, that is, $P$ is in general position.
A \emph{plane geometric matching} on $P$ is a plane geometric graph where each vertex is incident to at most one edge.
In the following, we refer to plane geometric matchings just as \emph{matchings}.
A matching on $P$ is called \emph{perfect} if each vertex is incident to exactly one edge, that is, the number of edges in the matching is $n$.

The concept of matchings has a long history of research, so here we survey only briefly some of the most recent results.
Sharir and Welzl~\cite{Sharir2006} provided bounds on the number of perfect matchings, all matchings (not necessarily perfect), and other variations of matchings that exist on a set $P$.
Aichholzer et al.~\cite{Aichholzer2009617} formulated the \emph{Disjoint Compatible Matching Conjecture} which was then proved by Ishaque et al.~\cite{Ishaque2011}:
For every perfect matching with an even number of edges there exists a disjoint compatible perfect matching.
In a slightly different direction, the compatibility of perfect matchings and different classes of plane geometric graphs is investigated.
In~\cite{CompatibleMatchings2011} it is shown that for outerplanar graphs there always exists a compatible perfect matching.
Further, upper and lower bounds are given on the number of edges shared between the given plane geometric graph and a compatible perfect matching, in case the graph is either a tree or a simple polygon.

Let $M$ and $M'$ be two perfect matchings on $P$.
According to~\cite{Aichholzer2009617} a \emph{transformation
  of length~$k$ between $M$ and $M'$} is a sequence of perfect
matchings $M=M_0,\ldots,M_k=M'$ such that $M_{i-1}$ and $M_{i}$ are
compatible for all $1\leq i\leq k$.
Let the \emph{transformation graph (of perfect matchings on $P$)} be the graph containing one node for each
perfect matching on $P$ and an edge joining two such nodes if and only
if the corresponding two perfect matchings are compatible, that is, there
exists a transformation of length~1 between these two perfect matchings.
Aichholzer et al.~\cite{Aichholzer2009617} proved that there
always exists a transformation of length~$O(\log n)$ between any two
matchings of $P$.
Hence, the transformation graph is connected with diameter~$O(\log n)$.
Providing a lower bound for the diameter, Razen~\cite{Razen2008} proved that there exist point sets $P$ such that the transformation graph (of~$P$) has diameter $\Omega(\log n/ \log \log n)$.

Given the wide interest in work on bichromatic point sets (see \cite{redbluesurvey} for a survey) it is only natural to extend the questions on matchings into that direction.
For the rest of this paper let $P= B\cup R$ be a bichromatic set of $2n$ points in the plane in general position, where $|B|=|R| = n$.
We call $B$ the set of blue points and $R$ the set of red points.
An edge of a geometric graph on $P$ is called \emph{bichromatic} if one endpoint of the edge is in $B$ and the other endpoint is in $R$.
A geometric graph is bichromatic, if all its edges are bichromatic.
For brevity, and in accordance with~\cite{bichromaticMatchings}, a perfect matching $M$ on $P$ is termed a \emph{$BR$-matching} if $M$ is bichromatic, that is, all edges of $M$ are bichromatic.

It is well known that a $BR$-matching always exists for any set $P$ as defined above.
For proofs see, e.g., \cite[p.~51]{larson183} (using the ``minimum weight is plane'' argument) and  \cite[pp.~200--201]{larson183} (using the intermediate value theorem).
On every set $P$ there also always exists a $BR$-matching constructed by repeated application of a ``ham-sandwich cut'' (see \figurename~\ref{fig:CanonicalMatching}).
We use such a $BR$-matching as the canonical structure (following the lines of~\cite{bichromaticMatchings}) and thus describe this in more detail in Section~\ref{section:Ham-Sandwich matchings}.
Concerning the maximal number of $BR$-matchings (over all sets $P$ with $|P|=2n$),  
Sharir and Welzl~\cite{Sharir2006} proved that it is at most $O(7.61^{2n})$ and can be bounded from below by $\Omega(2.23^{2n}/\poly(n))$ (where $\poly(n)$ stands for a polynomial factor in~$n$).

In a different direction, the augmentation of a disconnected bichromatic plane geometric graph with no isolated vertices to a connected bichromatic plane geometric graph has been considered.
The resulting connected (bichromatic) plane geometric graph is often called ``(bichromatic) encompassing graph''.
Hurtado et al.~\cite{Hurtado200814} proved that such an augmentation is always possible and provided an $O(n\log n)$ time algorithm to construct one.
This implies as a special case that every $BR$-matching can be augmented to a bichromatic plane spanning tree in $O(n\log n)$ time.
The result was extended by Hoffmann and T{\'o}th~\cite{HoffmannToth2013} to augmenting bichromatic geometric plane graphs
to bichromatic encompassing graphs where the increase of the degree of each vertex during the augmentation is bounded by two.
Thus, any $BR$-matching can be augmented to a bichromatic plane spanning tree with bounded degree three.
In a similar line of research Aichholzer et al.~\cite{CompatibleMatchingsForPSLG} proved that for every $BR$-matching there exists a bichromatic disjoint compatible matching $M'$ on $P$ with at least $\lceil\frac{n-1}{2}\rceil$ edges.
Furthermore, for an upper bound they provided an example where $M'$ has at most $3n/4$ edges.

Let $M$ and $M'$ be two $BR$-matchings.
Similar to the uncolored setting, a \emph{transformation of length~$k$ between $M$ and $M'$} is a sequence of $BR$-matchings $M=M_0,\ldots,M_k=M'$ such that $M_{i-1}$ and $M_{i}$ are compatible for all $1\leq i\leq k$.
The \emph{transformation graph} ${\cal M}_{BR}$ (of $BR$-matchings) is the graph containing one node for each $BR$-matching and an edge joining two such nodes if and only if the corresponding two $BR$-matchings are compatible.
Aloupis et al.~\cite{bichromaticMatchings} recently answered a question posed in~\cite{CompatibleMatchingsForPSLG}, proving that ${\cal M}_{BR}$ is connected for every point set $P=B\cup R$.
They presented a linear lower bound example for the maximum of the diameter of ${\cal M}_{BR}$ over all $P$. However, they provided no upper bound other than the trivial exponential bound stemming from the maximal number of nodes of ${\cal M}_{BR}$.

By adapting the approach and some of the tools presented in~\cite{bichromaticMatchings} we give an alternative proof of the connectivity of ${\cal M}_{BR}$.
A detailed analysis of each step of this proof allows us to prove %
an upper bound of $2n$ for the diameter of ${\cal M}_{BR}$. This is asymptotically tight, as there exist point sets $P$ for which ${\cal M}_{BR}$ has diameter $n/2$ (see \cite{bichromaticMatchings} and \figurename~\ref{fig:LowerBound}).

\section{The main result}\label{section:Ham-Sandwich matchings}

The main result of this paper is an asymptotically tight upper bound on the diameter of the transformation graph ${\cal M}_{BR}$ of $BR$-matchings,
derived by an alternative proof of the connectivity of ${\cal M}_{BR}$.
To this end, we define a canonical $BR$-matching and show that there exists a transformation of linear length between any $BR$-matching and the canonical one.

Throughout this paper, a \emph{ham-sandwich cut} of $P$ is a straight line $\ell$ such that (1) exactly $\lfloor \frac{n}{2}\rfloor$ blue and $\lfloor \frac{n}{2}\rfloor$ red points of $P$ are on one side of $\ell$ and~(2) exactly $\lceil \frac{n}{2}\rceil$ blue and $\lceil \frac{n}{2}\rceil$ red points of $P$ are on the other side of $\ell$, which implies that $\ell$ does not contain any point of $P$.
(Recall that we assume general position on $P$.)
For even $n$ this definition matches the ``classical'' definition for a ham-sandwich cut.
By the so-called \emph{Ham-sandwich Theorem} such a ham-sandwich cut always exists.
See~\cite{Borsuk33},~\cite{Goodman:1997:HDC:285869},~\cite{MatousekSteiger94}, and~\cite[Chapter~3]{MatousekBorsukUlam} for detailed information.
Furthermore, it is known that a ham-sandwich cut can be computed in $O(n)$ time~\cite{MatousekSteiger94}.
For odd $n$ a ``classical'' ham-sandwich cut $\ell_c$ of $P$ would contain a red and a blue point (on $\ell_c$).
We can shift $\ell_c$ slightly in parallel to achieve a ham-sandwich cut as defined above.

We construct a $BR$-matching $H$ by recursively applying ham-sandwich
cuts until in any cell there remain only two points, one of each
color, which are then matched (see \figurename~\ref{fig:CanonicalMatching}).
Recall that this is always possible by the Ham-sandwich Theorem.
In accordance with~\cite{bichromaticMatchings} we call $H$ a \emph{ham-sandwich matching}.
Note that several different ham-sandwich matchings might exist on $P$ and that, in general, not every $BR$-matching is a ham-sandwich matching.
Further, there exist point sets $P$ that admit only one single $BR$-matching, which then is a ham-sandwich matching.

\begin{figure}[tb]
\centering
\includegraphics{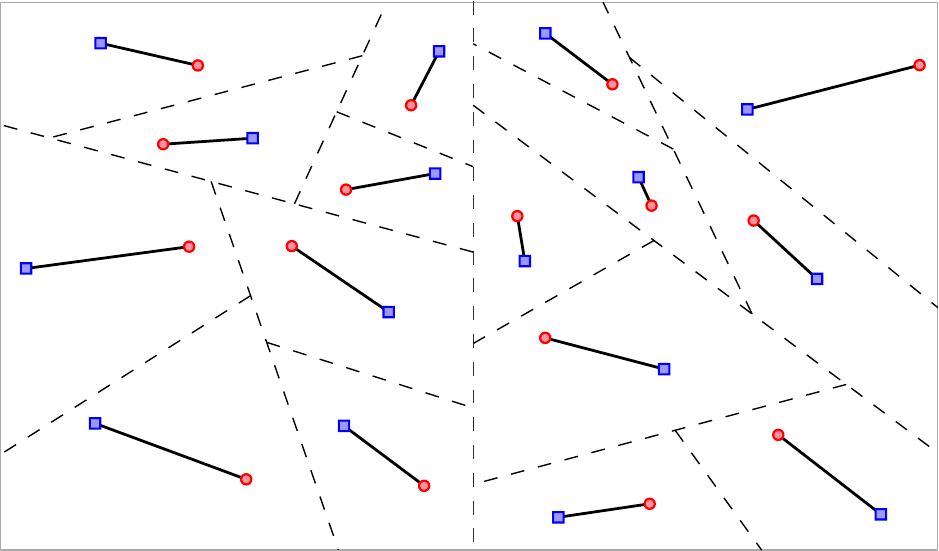}
\caption{\small \cite{bichromaticMatchings} A ham-sandwich matching obtained by repeated application of ham-sandwich cuts.
In our figures we depict blue points as filled squares and red points as filled disks.}
\label{fig:CanonicalMatching}
\end{figure}

One important ingredient for proving our main result (Theorem~\ref{theorem:The_graph_is_connected}) is Lemma~\ref{lemma:CompatibleSequence} stated below.
A similar result was obtained in~\cite{bichromaticMatchings} using comparable methods. However, that result did not permit to prove an upper bound on the diameter of ${\cal M}_{BR}$ (other than the trivial exponential one).
To not disrupt the train of thought we defer the proof of Lemma~\ref{lemma:CompatibleSequence} to Section~\ref{section:Avoiding_ham-sandwich_cuts}, as the remainder of this paper provides the tools for this proof.
Two $BR$-matchings $M$ and $M'$ are said to be \emph{$t$-compatible} if there exists a transformation of length~$k$ between $M$ and $M'$, with $k\leq t$.

\begin{lemma}\label{lemma:CompatibleSequence}
Let $P= B\cup R$ be a bichromatic set of $2n$ points in the plane in general position such that $|B|=|R|=n$. 
For every $BR$-matching $M$ and every ham-sandwich cut $\ell$ of $P$, there exists a $BR$-matching $M^\ell$ such that $M$ and $M^\ell$ are $\lfloor n/2 \rfloor$-compatible and no edge of $M^\ell$ intersects $\ell$.
\end{lemma}

Using this lemma, we obtain our main result.

\begin{theorem}\label{theorem:The_graph_is_connected}
Let $P= B\cup R$ be a bichromatic set of $2n$ points in the plane in general position such that $|B|=|R|=n$. 
For every $BR$-matching $M$ and every ham-sandwich matching $H$ of $P$, $M$ and $H$ are $n$-compatible.
\end{theorem}
\begin{proof}
We prove the statement by induction on $n$.
Trivially, the claim is true for $n=1$.
Hence, we proceed with the induction step and assume that the claim is true for any $1 \leq n'<n$.

Let $\ell$ be the first ham-sandwich cut in the construction of~$H$, i.e., a ham-sandwich cut of $P$.
By Lemma~\ref{lemma:CompatibleSequence}, there is a $BR$-matching $M^\ell$ such that $M$ and $M^\ell$ are $\lfloor n/2 \rfloor$-compatible and no edge of $M^\ell$ intersects $\ell$.
Let $P_1 = B_1 \cup R_1$ and $P_2 = B_2 \cup R_2$ %
be the subsets of points of $P$ lying to the left and to the right of $\ell$, respectively.
For each $i\in \{1,2\}$, let $M_i^\ell$ and $H_i$ be the subgraphs of $M^\ell$ and $H$, respectively, which are induced by $P_i$.
(Note that $H_1 \cup H_2 = H$ and $M_1^\ell \cup M_2^\ell = M^\ell$ as no edges of $M^\ell$ and $H$ intersect $\ell$.)

Let $\ell_1$ and $\ell_2$ be the ham-sandwich cuts of $P_1$ and $P_2$, respectively, used to construct $H$. 
Because $|P_i|=2n' \leq 2\lceil n/2 \rceil < 2n$, %
$M_i^\ell$ and $H_i$ are $\lceil n/2 \rceil$-compatible by induction.
Moreover, observe that every $B_1R_1$-matching is compatible with (and disjoint from) every $B_2R_2$-matching.
Thus, the two transformations of length~$k_i$ between $M_i^\ell$ and $H_i$ ($k_i\leq \lceil n/2 \rceil$) can be ``merged'' (i.e., executed in parallel) to one transformation of length~$\max_i\{k_i\}$ between $M^\ell$ and~$H$.
Finally, as $M$ and $M^\ell$ are $\lfloor n/2 \rfloor$-compatible and $M^\ell$ and $H$ are $\lceil n/2 \rceil$-compatible, we conclude that $M$ and $H$ are $n$-compatible.
\end{proof}

\begin{corollary}\label{corollary:upperbound}
Let $P= B\cup R$ be a bichromatic set of $2n$ points in the plane in general position such that $|B|=|R|=n$. 
The transformation graph ${\cal M}_{BR}$ is connected with diameter at most~$2n$.
\end{corollary}

\begin{figure}[htb]
\centering
\includegraphics{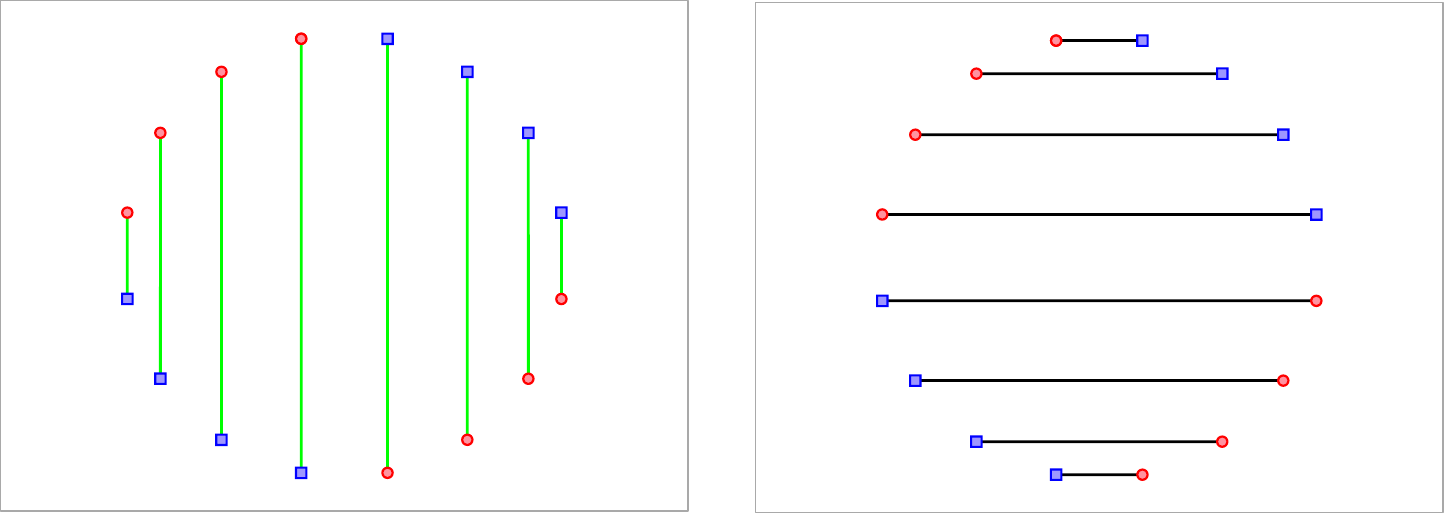}
\caption{\small \cite{bichromaticMatchings} Two ham-sandwich matchings that have distance $n/2$ in the transformation graph.}
\label{fig:LowerBound}
\end{figure}

The example depicted in \figurename~\ref{fig:LowerBound}, which has also been presented in~\cite{bichromaticMatchings}, shows that the diameter of the transformation graph ${\cal M}_{BR}$ can be as high as $n/2$. Together with Corollary~\ref{corollary:upperbound}, we obtain the following result.

\begin{corollary}\label{corollary:G_P_connected}
The maximum over all bichromatic sets $P= B\cup R$ with $|B|=|R|=n$ of the diameter of the transformation graph ${\cal M}_{BR}$ is $\Theta(n)$.
\end{corollary}

Note that the lower bound for the diameter of ${\cal M}_{BR}$ is $0$, as there exist point sets $P= B\cup R$ with $|B|=|R|=n$ admitting only one $BR$-matching.

\section{Proof of Lemma~\ref{lemma:CompatibleSequence}}

For the remainder of this paper, we consider each edge of a plane geometric graph $G$ %
to have two sides.
Formally, each edge $pq$ of $G$ consists of a pair of \emph{\hedges}, one directed from $p$ to $q$ and the other directed from $q$ to $p$ such that the cycle of each \hedge\ pair is oriented clockwise (see \figurename~\ref{fig:ColoringTheBoundary}~(a)).
Each \hedge\ is colored either red or blue. 
For an edge $pq$ the \hedge\ directed to $p$ is called the \emph{twin} of the \hedge\ directed to $q$, and vice versa.
Let $\ell_{pq}$ be the line supporting the edge $pq$, and being directed from $p$ to $q$.
Only the \hedge\ directed to $q$ is visible from the left side of $\ell_{pq}$ whereas only the \hedge\ directed to $p$ is visible from the right side of~$\ell_{pq}$.
In other words, a \hedge\ is visible only from its left side and has its twin on its right side.
Note that a point $x$ on an open edge $pq$ with differently colored \hedges\ is observed as being red from one side of $\ell_{pq}$, while $x$ appears to be blue from the other side of $\ell_{pq}$ (see again \figurename~\ref{fig:ColoringTheBoundary}~(a)).

\begin{figure}[htb]
\centering
\includegraphics{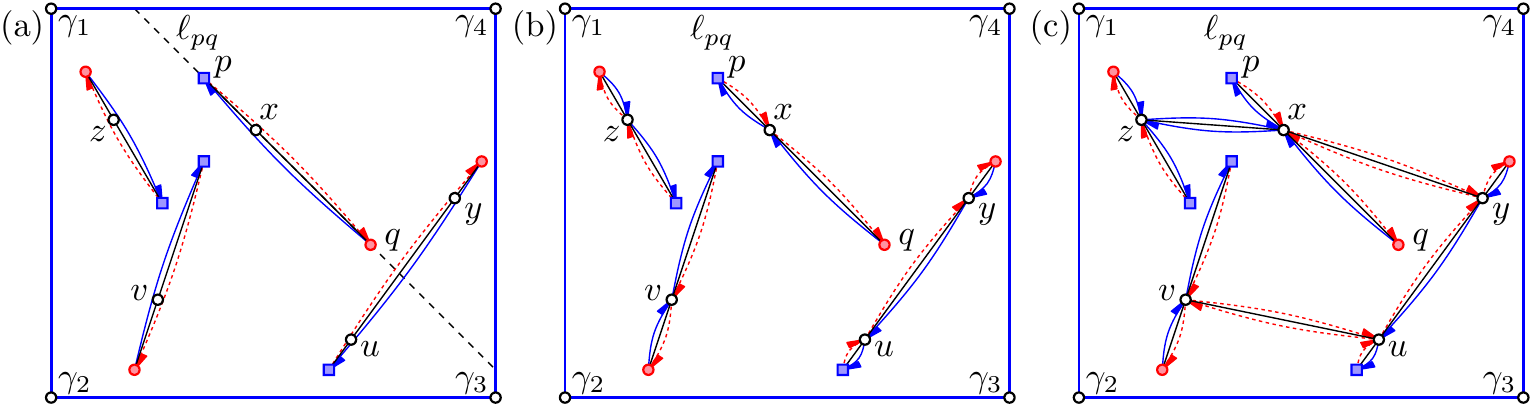}
\caption{\small Splitting and gluing in $P$-graphs: In the figures we show points of $Q$ as white disks, blue \hedges\ as solid arcs, and red \hedges\ as dotted arcs.
The points $\{\gamma_1,\ldots,\gamma_4\}$ are the vertices of the rectangle $\Gamma$.
For simplicity $\Gamma$ is displayed with bold lines instead of equally-colored \hedges.
(a) Each edge $pq$ of $M$ has two \hedges, one \hedge\
  directed to $p$ and colored like $p$, the other directed to $q$ and
  colored like $q$. The points $z$ and $u$ are not visible, $x$ and
  $u$ are visible but not \cvisible, and $x$ and $y$ are \cvisible.
(b) The resulting graph when splitting $pq$ at $x$ and the other edges
at $z$, $v$, $u$, and $y$.
(c) The resulting graph after gluing three pairs of \cvisible{} points
$x,y$, $x,z$, and $u,v$.}
\label{fig:ColoringTheBoundary}
\end{figure}

Let $M$ be a $BR$-matching. %
For each edge $s$ of $M$ color the \hedges\ of $s$ in the same color as the endpoint towards which they are directed to.
In this way, every edge of a $BR$-matching %
has a blue \hedge\ and a red \hedge. Moreover, this coloring is uniquely determined by $P$ (and the fixed orientation of \hedge\ pairs).
Let $\conv$ be an axis aligned rectangle sufficiently large to enclose $M$ in its interior.
We color each \hedge\ on the boundary of $\conv$ with the same color (to be determined later).
See \figurename~\ref{fig:ColoringTheBoundary}~(a) for an illustration where each \hedge\ of $\conv$ is colored blue.

We define a \emph{$P$-graph} (of $M$ and $\Gamma$) to be a plane geometric graph $G_M$ on a point set $P\cup Q$ such that
(1) $Q$ is disjoint from $P$, %
(2) $G_M$ contains a subdivision of $\conv$ and a subdivision of $M$ as subgraphs, 
(3) for every edge of $M$ its \hedges\ are colored as defined above, and
(4) for every edge of $G_M$ that is not an edge of $M$, its two \hedges\ are colored in the same color, either red or blue.
(We do not require $P\cup Q$ to be in general position, but recall that we assume general position of $P$.)
From now on we only consider the part of the plane bounded by $\conv$.
Thus, each considered face $f$ of $G_M$ is bounded.
We denote by $\partial f$ the boundary of~$f$ and by $\fint(f)$ the interior of~$f$.
Furthermore, let the \emph{boundary} of $G_M$, denoted by $\partial G_M$, be the union of all the edges in $G_M$, and let the \emph{interior} of $G_M$ be the union of the interiors of its faces.

Consider two points $x$ and $y$ that lie on different edges of $\partial G_M$.
We say that $x$ and $y$ are \emph{visible} if the open segment joining $x$ with $y$ is contained in the interior of $G_M$.
We say that $x$ and $y$ are \emph{\cvisible} if they are visible and the color of $x$ when viewed from $y$ is equal to the color of $y$ when viewed from $x$.
For example, in \figurename~\ref{fig:ColoringTheBoundary}~(a), $u$ and $x$ are visible but not \cvisible, while $x$ and $y$ are \cvisible.

With these definitions, we first show how to create a $P$-graph of $M$ that is a convex decomposition of the interior of $\conv$.
To this end we define the glue operation, as has been done in~\cite{bichromaticMatchings}, and use a colored version of an extension of a matching (see e.g.~\cite{Aichholzer2009617} for uncolored extension).
Then we show how to construct a $BR$-matching that is compatible to the created convex decomposition and prove that this $BR$-matching has strictly less intersections with a ham-sandwich cut of~$P$ than~$M$.

\subsection{Splitting and gluing in $P$-graphs}

Consider a $P$-graph $G_M$ on $P\cup Q$ and let $x\notin P\cup Q$ be a point on an edge $pq$ of $G_M$.
To \emph{split} $pq$ at $x$ we do the following:
(1) add $x$ to $Q$,
(2) add the edges $px$ and $xq$ to $G_M$,
(3) color the \hedges\ from $p$ to $x$ and from $x$ to $q$ like the
\hedge\ from $p$ to $q$, and the other two new \hedges\ like the \hedge\
from $q$ to $p$, and
(4) remove $pq$ (and its two \hedges) from $G_M$.
\figurename~\ref{fig:ColoringTheBoundary}~(a-b) gives an illustration of the split operation.

We borrow the gluing technique introduced in~\cite{bichromaticMatchings}:
Let $y$ and $y'$ be two \cvisible{} points on two different edges $e$ and $e'$, respectively, of $\partial G_M$ such that neither $y$ nor $y'$ is in $P$.
To \emph{glue} $y$ with~$y'$, we do the following:
If $y$ (or $y'$) is not a vertex of $G_M$, then we split $e$ at $y$ (or $e'$ at $y'$), by this ensuring that $y$ and $y'$ are now vertices of $G_M$.
Then we add the edge $yy'$ to $G_M$ and color the two \hedges\ of $yy'$ with the same color as $y$ when viewed from $y'$.
See \figurename~\ref{fig:ColoringTheBoundary}~(b-c) for examples of gluing.

\begin{observation}
\label{obs:slitandglue}
The resulting graph of splitting an edge of a $P$-graph at a point on this edge is again a $P$-graph.
The resulting graph of gluing two \cvisible{} points (neither of them in $P$) on two different edges of a $P$-graph is again a $P$-graph.
\end{observation}

\begin{figure}[tb]
\centering
\includegraphics{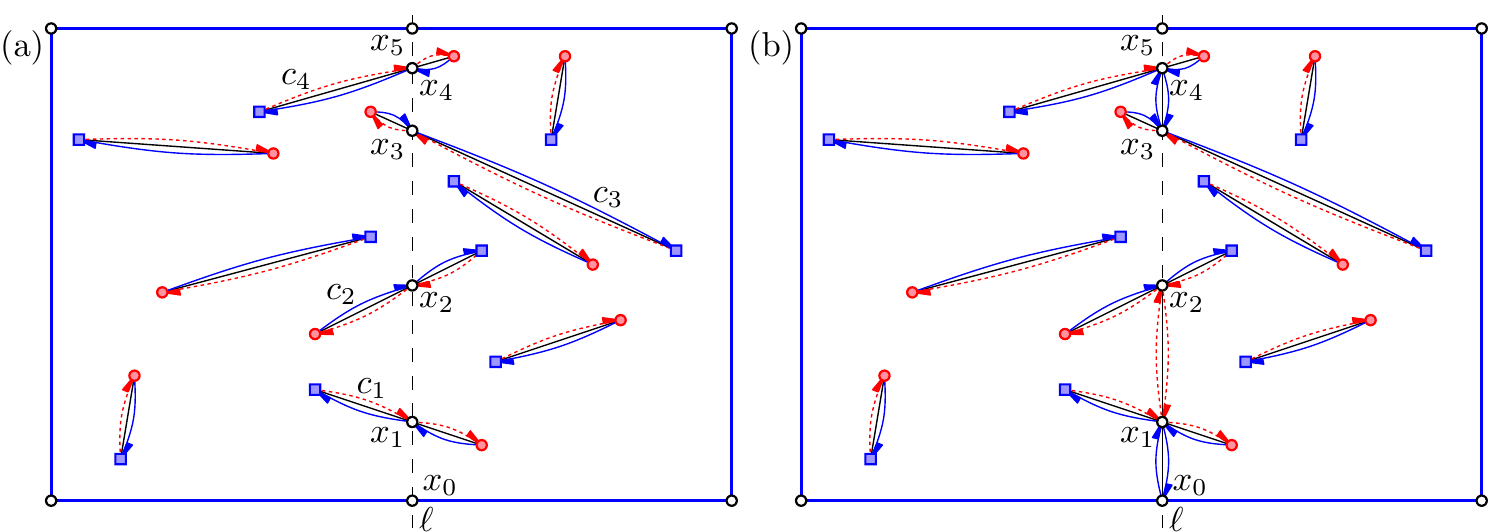}
\caption{\small Generating $G_M^0$: (a) The edges $\langle c_1,\ldots,c_4\rangle$ of $M$ intersect $\ell$ in $\langle x_1,\ldots,x_4\rangle$ and are split at these points. 
(b) The pairs of \cvisible{} points $x_0,x_1$, $x_1,x_2$, and $x_3,x_4$ are glued.}
\label{fig:Gluing the segments}
\end{figure}

Consider a $P$-graph $G_M$ on $P\cup Q$ with $Q$ only containing the four points of $\conv$ and $G_M$ containing only the edges of $M$ and $\conv$.
Let $\ell$ be a ham-sandwich cut of $P$ and assume without loss of generality that $\ell$ is vertical and that no edge of $M$ is parallel to~$\ell$.
Let $C_{M, \ell} = \langle c_1, \ldots, c_k\rangle$ be the sequence of $k$ edges of $M$ that intersect $\ell$, sorted from bottom to top according to the point of intersection $x_i$ of $c_i$ with $\ell$.
Let $x_0$ and $x_{k+1}$ be the intersection points of $\ell$ with the bottom edge and top edge of $\conv$, respectively.
Color each \hedge\ on the boundary of $\conv$ with the same color as $x_1$ when viewed from $x_0$; see \figurename~\ref{fig:Gluing the segments}~(a).
Recall that Lemma~\ref{lemma:CompatibleSequence} looks for a $BR$-matching $M^\ell$, such that $M$ and $M^\ell$ are compatible and $M^\ell$ has no edges intersecting $\ell$.
Therefore, we can assume that $k > 0$ as otherwise we have already found the desired $BR$-matching.
We construct a $P$-graph $G_M^0$ by gluing $x_i$ with $x_{i+1}$, for each $0\leq i\leq k$, if $x_i$ and $x_{i+1}$ are \cvisible. 
By doing so, we ensure that no edge in a $BR$-matching compatible with $G_M^0$ can intersect $\ell$ between $x_i$ and $x_{i+1}$, if $x_i$ and $x_{i+1}$ are \cvisible.
Recall that the \hedges\ on $\conv$ have the color of $c_1$ when viewed from below. 
That is, the points $x_0$ and $x_1$ are \cvisible{} and hence, they are glued together; see \figurename~\ref{fig:Gluing the segments}~(b) for an illustration.

\begin{observation}\label{obs:Consecutive_points_have_different_colors}
Let $M$ be any $BR$-matching on $P$ and let $\ell$ be any ham-sandwich cut of $P$, such that the intersection of $\ell$ with the edges of $M$ is not empty.
There exists a $P$-graph $G_M^0$ such that two points $x_i$ and $x_{i+1}$ are joined by an edge in $G_M^0$ if and only if $x_i$ and $x_{i+1}$ are \cvisible.
Moreover, $x_0$ and $x_1$ are always glued by an edge of $G_M^0$.
\end{observation}

\subsection{Extension of $M$}

In this section, we describe the extension of the $BR$-matching $M$ in the $P$-graph $G_M^0$.
Let $s_1, \ldots, s_n$ be an arbitrary order of the edges of $M$.
Starting with $G_M^0$, we extend each edge of $M$ in this order, resulting in a sequence $G_M^0, \ldots, G_M^n$ of $P$-graphs.

During this sequence we maintain the following \emph{\invariant}:
For $0\leq j\leq n$ and every pair of points $u,v \in (\partial G_M^j\cap\ell)$, $u$ and $v$ are not \cvisible.
Intuitively, the \invariant{} guarantees that every interval along $\ell$ that is not covered by an edge of $G_M^j$ is bounded by points having different colors.

\begin{lemma}\label{lemma:Color_invariant_holds_at_the_beginning}
The \invariant{} holds for $G_M^0$.
\end{lemma}
\begin{proof}
Recall that $x_0$ and $x_{k+1}$ are the intersections of $\ell$ with $\conv$ and that for every $1\leq i\leq k$, $x_i$ is the intersection of the edge $c_i\in C_{M,\ell}$ with the line $\ell$.
For two points to be \cvisible{} they need to be visible.
In $\partial G_M^0\cap\ell$ only the points $x_i$ and $x_{i+1}$, for some $0\leq i\leq k$, can be visible. 
By Observation~\ref{obs:Consecutive_points_have_different_colors}, $x_i$ and $x_{i+1}$ are visible in $G_M^0$ if and only if they are not \cvisible{} in $M$, i.e., the \invariant{} holds.
\end{proof}

We proceed by describing the extension of $M$ in detail.
For each edge $s_j$ of $M$ the \emph{extension} of $s_j$ comes in three
steps:
(1) shooting a ray from $s_j$ to both directions until hitting an edge of $G_M^{j-1}$,
(2) proper coloring of the \hedges\ of the two rays, and
(3) maintaining the \invariant.

\noindent\emph{Step 1:} Let $\ell_s$ be the supporting line of $s_j=pq$.
Let $z_p$ and $z_q$ be the intersection points of $\ell_s$ and $\partial G_M^{j-1}$, such that $p$ and $z_p$ are visible and $q$ and $z_q$ are visible.
Note that such an intersection can be with an edge of $M$, with an edge of $\conv$, or with any other edge of $\partial G_M^{j-1}$.
If any of $z_p$ or $z_q$ is not a vertex of $G_M^{j-1}$ then split the edge containing $z_p$ at $z_p$ or split the edge containing $z_q$ at $z_q$, respectively.
Extend $s_j$ by adding the edges $pz_p$ and $qz_q$ to $G_M^{j-1}$.
See \figurename~\ref{fig:Extensions}~(a) for an example.

\noindent\emph{Step 2:} The two \hedges\ of $pz_p$ are colored with the same
color as $z_p$ when viewed from $p$.
The two \hedges\ of $qz_q$ are colored with the same color as $z_q$ when viewed from $q$; see \figurename~\ref{fig:Extensions}~(b).
By this coloring, the resulting graph is a $P$-graph.

\noindent\emph{Step 3:} Observe that at most one of the two new edges
can intersect $\ell$.
Assume that the \invariant{} holds before processing $s_j$. 
If neither of the two new edges intersects $\ell$, then the \invariant{} still holds after extending $s_j$.
Thus, without loss of generality, assume that $pz_p$ intersects $\ell$ in point $y$.
Let $z_u$ and $z_d$ be first points hit on $\partial G_M^{j-1}$ when shooting upwards and downwards, respectively, from $y$ along $\ell$. 
The \invariant{} guarantees that $z_u$ and $z_d$ are not \cvisible{} in $G_M^{j-1}$.
Hence, $z_u$ and $z_d$ have different colors when viewed from $y$, but $y$ has the same color independent of being viewed from $z_u$ or $z_d$.
Therefore, $y$ and exactly one of the two points $z_u$ and $z_d$ are \cvisible.
We glue $y$ with this \cvisible{} point;  %
see \figurename~\ref{fig:Extensions}~(b) for an example.

\begin{figure}[tb]
\centering
\includegraphics{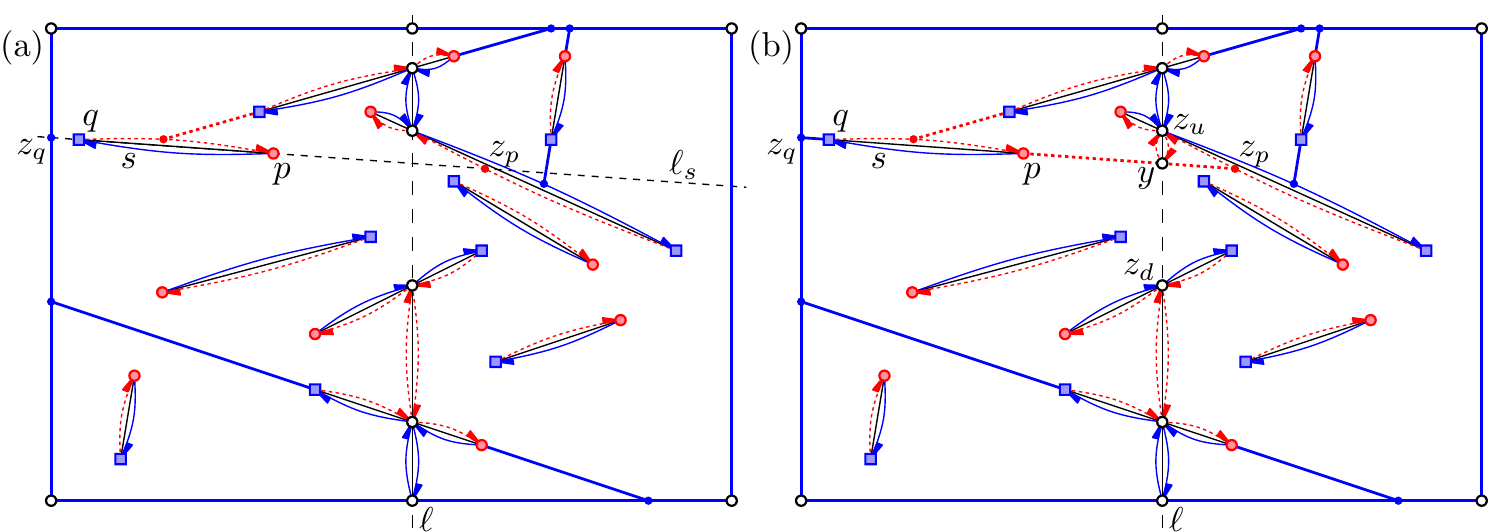}
\caption{\small Extending $M$ in $G_M^0$: For simplification, extensions of edges are displayed with bold lines instead of
equally-colored \hedges\ and endpoints of extensions are depicted as
small dots.
(a) The edge $s=pq$ is extended on its supporting line $\ell_s$, hitting the blue point $z_q$ and the red point $z_p$.
(b) As $z_q$ is blue and $z_p$ is red, the extensions $qz_q$ and $pz_p$ are blue and red, respectively. 
The extension $pz_p$ intersects $\ell$ in the point $y$ inside the interval $z_dz_u$ on $\ell$.
As $y$ and $z_u$ are \cvisible{} (red), they are glued.}
\label{fig:Extensions}
\end{figure}

\begin{lemma}\label{lemma:Color_invariant_preserved}
The \invariant{} is preserved after each extension of an edge of~$M$.
In particular, this invariant holds in the resulting graph $G_M^n$, after extending every edge of~$M$.
\end{lemma}
\begin{proof}
By Lemma~\ref{lemma:Color_invariant_holds_at_the_beginning}, the \invariant{} holds for $G_M^0$ before extending $s_1$.
We prove by induction and thus assume that the \invariant{} is preserved until extending $s_j$.
Observe that the \invariant{} can only be violated if a new edge (at most one of the two extensions of $s_j$) intersects~$\ell$ in a point~$y$.
If this is the case then $y$ lies between two points $z_d$ and $z_u$ that are visible in~$G_M^{j-1}$.
As argued above, $y$ and exactly one of the two points, without loss of generality $z_d$, are \cvisible{} in~$G_M^{j-1}$.
As $y$ is glued with $z_d$ in Step 3, $y$ and $z_d$ are not visible in $G_M^{j}$.
Furthermore, all other pairs of visible points of $\partial G_M^{j-1}$ on $\ell$ remain unchanged. Thus, the \invariant{} also holds after extending~$s_j$.
\end{proof}

It is easy to see that the resulting $P$-graph $G_M^n$ decomposes the interior of $\conv$ into convex simple polygons, each being a face of $G_M^n$; see \figurename~\ref{fig:Switch_vertices}~(a).
Note that every point in the interior of each face of $G_M^n$ sees a counterclockwise directed cycle of colored \hedges.

In the following two sections we construct a $BR$-matching $M'$ compatible to $G_M^n$.
Recall that the edges of $M'$ should have as few intersections with the ham-sandwich cut $\ell$ as possible.
As $G_M^n$ and $M'$ are compatible, only edges inside a face of $G_M^n$ can intersect $\ell$.
Thus, we are interested in the number of faces of $G_M^n$ that contain a portion of $\ell$ in their interior.
We say that a face $f$ of $G$ \emph{crosses} $\ell$ if $\fint(f)\cap\ell\neq\emptyset$.

\begin{lemma}\label{lemma:Few_faces_crossing_ell}
At most $k-1$ faces of $G_M^n$ cross $\ell$, where $k = |C_{M, \ell}|$.
\end{lemma}
\begin{proof}
Recall that $C_{M,\ell} = \langle c_1, \ldots, c_k\rangle$ is the sequence of edges of $M$ that intersect $\ell$ and that for every $1\leq i\leq k$, $x_i$ is the intersection point of $c_i$ with $\ell$.
Further recall that $x_0$ and $x_{k+1}$ are the intersections of $\ell$ with $\conv$ and that we assume that $k>0$, as otherwise $M$ would already fulfill the requirements of Lemma~\ref{lemma:CompatibleSequence}.
In~\cite{bichromaticMatchings} it was already observed that if $\ell$ intersects at least one edge of $M$, then it must intersect an even number of edges of $M$.
Moreover, as $\ell$ is a ham-sandwich cut, at each side of $\ell$ the number of red points equals the number of blue points.
Therefore, if we consider the endpoints of the edges in $C_{M, \ell}$ at one side of $\ell$, half of them must be blue and half must be red.
Otherwise, the numbers of remaining red and blue points at that side of $\ell$ would be unbalanced, leading to a contradiction with $M$ being a $BR$-matching. 
Thus, there exists at least one $\xi\in\{1,\ldots,k-1\}$ such that the pair of consecutive edges $c_\xi$ and $c_{\xi+1}$ in $C_{M, \ell}$ has differently colored endpoints at the same side of $\ell$.
By the coloring scheme of the \hedges\ of $M$, $x_\xi$ and $x_{\xi+1}$ are \cvisible{} in $M$.

For $0\leq j\leq n$, let $\omega_j$ be the number of connected components of $\ell \setminus \partial G_M^j$ that lie inside $\conv$.
Observe that inside $\conv$ the number of connected components of $\ell$ intersected by the edges of $M$ is $k+1$.
By Observation~\ref{obs:Consecutive_points_have_different_colors}, $x_i$ is glued with $x_{i+1}$ in the construction of $G_M^0$ if and only if $x_i$ and $x_{i+1}$ are \cvisible.
By the choice of the color of the \hedges\ of $\conv$, $x_0$ is glued with $x_1$.
As argued above, there exists at least one additional pair $x_\xi$ and $x_{\xi+1}$ that is \cvisible{} and thus glued in $G_M^0$.
Hence, $\omega_0$ is at most $k-1$.

In the construction of $G_M^j$, $1\leq j\leq n$, the connected components of $(\ell \setminus \partial G_M^{j-1})\cap\conv$ remain unchanged unless exactly one new edge intersects~$\ell$.
In this case, exactly one connected component gets split into two connected components, of which exactly one connected component is removed in $G_M^{j}$ by gluing its endpoints.
Thus, $\omega_{j}=\omega_{j-1}$ for all $1\leq j\leq n$.

As the faces of $G_M^{n}$ are convex simple polygons, the number of faces of $G_M^{n}$ that cross $\ell$ is equal to~$\omega_{n}$ and thus at most $k-1$.
\end{proof}

\subsection{Switch vertices and switch matchings}

Recall that $G_M^{n}$ decomposes the interior of $\conv$ into convex faces. 
The idea is to assign each point of $P$ to a unique face of this decomposition, such that every face has a balanced number %
of (possibly zero) red and blue points assigned.
This way, we obtain a new $BR$-matching by independently matching the points assigned to each face of this decomposition.

Note that each \hedge\ of $G_M^{n}$ is incident to the interior of a unique face $f$ of $G_M^{n}$.
Therefore, we can think of $\partial f$ to be composed of all the \hedges\ incident to~$\fint(f)$.
Consider the sequence $h_0,\ldots,h_{t-1}$ of the $t\geq 3$ \hedges\ along $\partial f$ in counterclockwise order, i.e., the cycle formed of the $t$ \hedges\ incident to~$\fint(f)$. 

A vertex $v$ of $G_M^{n}$ is a \emph{\svertex} in $f$ if the two \hedges\ $h_i$ and $h_{i+1}$ (with $i\in\{0\ldots t-1\}$ and indices taken modulo $t$) that are incident to~$\fint(f)$ and adjacent to $v$ have different colors; see \figurename~\ref{fig:Switch_vertices}~(b) for an illustration. 
In other words, $v$ is \svertex{} in some face if, in the cyclic order of incident \hedges\ around $v$, two consecutive \hedges\ that are not twins have different color.

\begin{figure}[tb]
\centering
\includegraphics{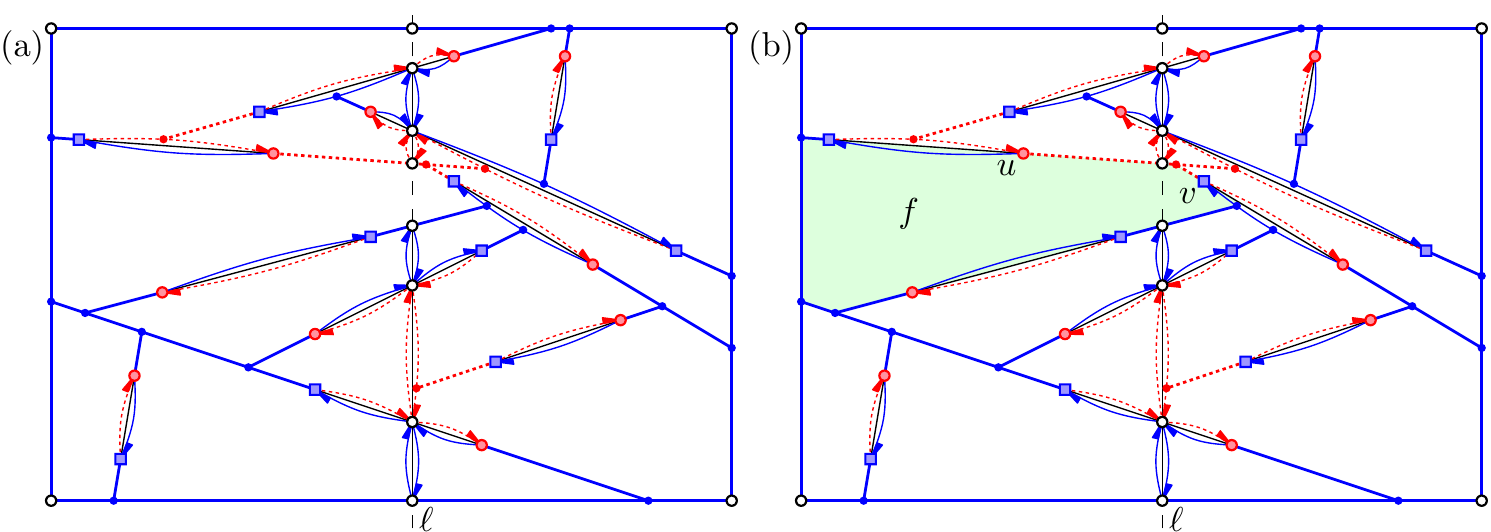}
\caption{\small (a) The convex decomposition of $\conv$ obtained after extending every edge of $M$.
(b) The vertices $u$ and $v$ are the only \svertices{} in the face $f$.}
\label{fig:Switch_vertices}
\end{figure}

\begin{lemma}\label{lemma:Properties_of_switch_vertices}
A vertex of $G_M^{n}$ is a \svertex{} in one of its faces if and only if it is a point of~$P$.
Furthermore, a vertex can be a \svertex{} in at most one face of $G_M^{n}$.
\end{lemma}
\begin{proof}
For each point $q$ of $G_M^{n}$ let $\Delta_q$ be the cyclic order of its incident \hedges.

First observe that splitting an edge of a $P$-graph at a point $x$ preserves $\Delta_q$ for all $q\in P\cup Q\setminus \{x\}$.
(Strictly speaking, at the endpoints of the split edge the split \hedges\ get exchanged with the new \hedges.
But as their color stays the same, the cyclic order of the colors of \hedges\ around these points stays the same.)
Further, for the new point $x\in Q$, $\Delta_x$ contains two pairs of consecutive \hedges\ that are not twins, and both pairs consist of equally-colored \hedges.
Hence, the split operation preserves existing \svertices{} and does not create new ones.

Second, let $y\in Q$ be a point that is glued with another point in $Q$.
This means that two equally-colored \hedges\ are inserted between two equally-colored \hedges\ of the same color in $\Delta_y$.
Therefore, no point in $Q$ becomes a \svertex{} by the glue operation.

Third, let $z$ be the point on some edge of the $P$-graph that is first hit by the extension of one side of some edge of $M$.
If not already in $Q$, $z$ gets added to $Q$ by a split operation.
Then, like in the glue operation, two equally-colored \hedges\ are inserted between two equally-colored \hedges\ of the same color in $\Delta_{z}$.
Again, no point in $Q$ becomes a \svertex{} by this operation.

Altogether, no point of $Q$ is turned into a \svertex{} during the construction of $G_M^{n}$.
Further, all points in $Q$ are either points of $\conv$ (whose incident \hedges\ are all of the same color) or created in a split operation.
Therefore, no point in $Q$ is a \svertex.

Concerning the set $P$ recall that each point $p\in P$ is an endpoint of an edge $s=pp'$ of $M$.
As argued above, only the extension of $s$ alters $\Delta_p$ during the construction of $G_M^{n}$.
Before the extension of $s$, each of the endpoints $p$ and $p'$ of $s$ is incident to exactly one twin pair of \hedges\ (where one is colored red and the other one is colored blue).
The extension of $s$ adds two additional \hedges\ to $p$, both of the same color.
Thus, $\Delta_p$ has exactly two pairs of consecutive \hedges\ that are not twins, and for exactly one of them the two \hedges\ differ in color. The same statement holds for $\Delta_{p'}$. 
Note that for all points \mbox{$\tilde{p}\in P\!\setminus\!\{p,p'\}$} this operation preserves $\Delta_{\tilde{p}}$.
Therefore, every point in $P$ is a \svertex{} for exactly one face of $G_M^{n}$.
\end{proof}

\begin{lemma}\label{lemma:Clockwise_coloring}
Let $h_0, \ldots, h_{t-1}$ be the sequence of \hedges\ along the boundary of a face $f$ of $G_M^{n}$ in counterclockwise order.
Let $v_i$ be a \svertex{} in $f$ and let $h_i$ and $h_{i+1}$ (indices taken modulo~$t$) be the two \hedges\ incident to $v_i$.
Then $v_i$ has the same color as $h_i$ while $h_{i+1}$ is of the opposite color.
\end{lemma}
\begin{proof}
By Lemma~\ref{lemma:Properties_of_switch_vertices}, $v_i$ is a point of $P$.
Hence, $v_i$ is the endpoint of an edge~$s$ of~$M$.
Let $s'$ be the part of $s$ (after possible splits) incident to $v_i$ in $G_M^{n}$.
Recall that splitting an edge of a $P$-graph preserves the cyclic order of incident \hedges\ for all points in $P$.
Therefore, the \hedge\ $h^+$ of~$s'$ directed towards $v_i$ has the same
color as $v_i$, and the \hedge\ $h^-$ of $s'$ directed away from~$v_i$ has the
opposite color of $v_i$.

In case that $h_i$ is $h^+$, $h_i$ has the same color as $v_i$ and, since $v_i$ is a \svertex, $h_{i+1}$ must be of the opposite color.
In the other case, where $h_{i+1}$ is $h^-$, $h_{i+1}$ is of the opposite color as $v_i$ and as $v_i$ is a \svertex, $h_i$ must have the same color as $v_i$.
Thus, in both cases the claim in the lemma is true.
\end{proof}

We say that a face $f$ of $G_M^{n}$ is \emph{\wc} if the sequence of \svertices{} along $\partial f$ alternates in color.
Analogously, a $P$-graph is \emph{\wc} if all its faces are \wc. 
Notice that if a face is \wc, then it has an even number of \svertices.

\begin{lemma}\label{lemma:Well-colored faces}
Every face of $G_M^{n}$ is \wc.
\end{lemma}
\begin{proof}
Let $h_0, \ldots, h_{t-1}$ be the sequence of \hedges\ along the boundary of a face $f$ of $G_M^{n}$ in counterclockwise order.
For any $0\leq i \leq t-1$, let $v_i$ be the vertex shared by $h_i$ and $h_{i+1}$ (indices taken modulo $t$). 
Recall that $h_i$ and $h_{i+1}$ have different colors if and only if~$v_i$ is a \svertex{} in~$f$.

Let $v_i$ and $v_j$ be two consecutive \svertices{} along $\partial f$ such that~$i<j<t$.
Assume without loss of generality that $v_i$ is red.
Therefore, Lemma~\ref{lemma:Clockwise_coloring} implies that $h_i$ is red whereas $h_{i+1}$ is blue. 
Because $v_i$ and $v_j$ are consecutive \svertices{} along $\partial f$, for every $i < r < j$, $v_r$ is not a \svertex. 
Thus, $h_{i+1}, \ldots, h_{j}$ share the same color, i.e., they are blue.
Because $v_j$ is a \svertex, $h_j$ and $h_{j+1}$ have different colors, which implies that $h_{j+1}$ is red.
Since $h_j$ is blue and $h_{j+1}$ is red, we infer from Lemma~\ref{lemma:Clockwise_coloring} that $v_j$ is blue. 
Therefore, $v_i$ and $v_j$ have different colors, i.e., two consecutive \svertices{} along $\partial f$ alternate in color, which implies that $f$ is \wc.
\end{proof}

Let $f$ be a \wc{} face of $G_M^{n}$ and let $P_f$ be the set of \svertices{} of $f$.
A \emph{\smatching} $M_f$ of $f$ is a $BR$-matching on $P_f$ such that every edge of $M_f$ is contained in $f$ (or on $\partial f$).
Since $f$ is \wc, the sequence of \svertices{} along $\partial f$ alternates in color.
Moreover, since $f$ is a convex simple polygon, we can obtain $M_f$ by connecting consecutive \svertices{} along $\partial f$.
That is, every face of $G_M^{n}$ admits a \smatching.

Recall that a vertex is a \svertex{} in exactly one face of $G_M^{n}$ by Lemma~\ref{lemma:Properties_of_switch_vertices}.
Therefore, as every face of $G_M^{n}$ is \wc{} by Lemma~\ref{lemma:Well-colored faces}, we can obtain a $BR$-matching compatible with $M$ by taking the union of the \smatchings{} of every face in $G_M^{n}$. 
However, this $BR$-matching may have more crossings with $\ell$ than $M$, so we need to be careful when matching the \svertices{} of~$G_M^{n}$.

\begin{figure}[tb]
\centering
\includegraphics{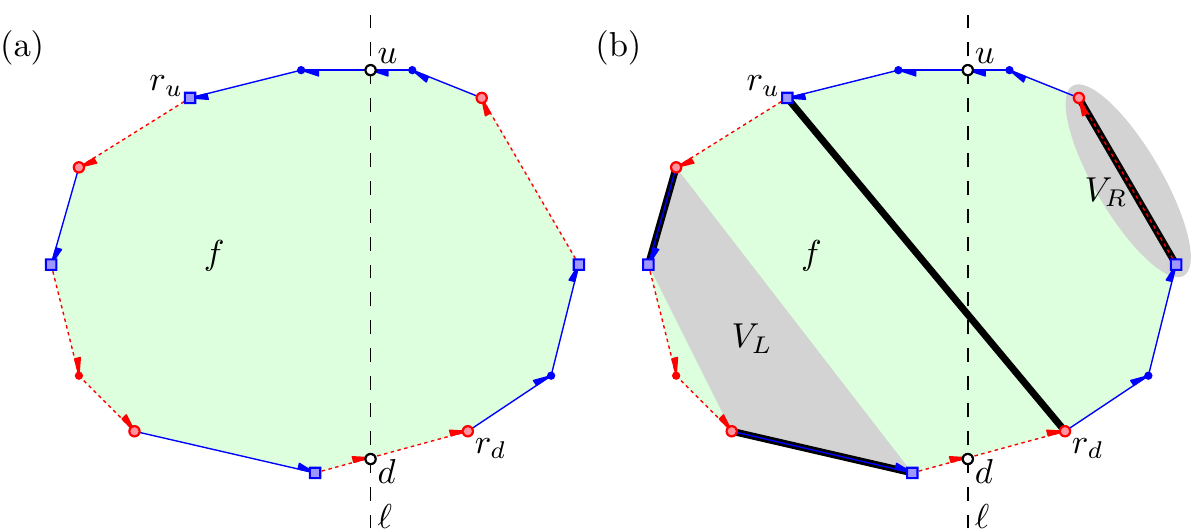}
\caption{\small (a) A \wc{} face $f$ of $G_M^{n}$ whose boundary intersects $\ell$ at points $u$ and $d$.
The vertices $r_u$ and $r_d$ are the first \svertices{} encountered when walking along the boundary of~$f$ counterclockwise from $u$ and $d$, respectively.
(b) The sets $V_L$ and $V_R$ contain the \svertices{} of $f$ lying to the left and right, respectively, of $r_dr_u$.
Moreover, the convex hulls of $V_L$ and $V_R$ are contained to the left and right, respectively, of $\ell$.
($V_R$ contains only two points and thus its convex hull has no area.)
The bold edges exemplify one \smatching{} $M_f$.}
\label{fig:Switch-matchings}
\end{figure}

\begin{lemma}\label{lemma:Switch-matchings}
Let $f$ be a \wc{} face of $G_M^{n}$ that crosses $\ell$.
There exists a \smatching{} $M_f$ on the \svertices{} of $f$ such that at most one edge of $M_f$ intersects $\ell$.
\end{lemma}
\begin{proof}
Since $f$ is a convex polygon, $\ell$ intersects $\partial f$ in exactly two points $u$ and $d$.
Assume without loss of generality that $u$ lies above $d$; see \figurename~\ref{fig:Switch-matchings}~(a).
Notice that $u$ and $d$ are visible points in $G_M^{n}$ lying on the line $\ell$. 
Because the \invariant{} holds in $G_M^{n}$ by Lemma~\ref{lemma:Color_invariant_preserved}, $u$ and $d$ are not \cvisible.
So, without loss of generality, assume that $u$ is blue when viewed from $d$ and hence that $d$ is red when viewed from $u$.
Walk counterclockwise from $u$ and $d$ along $\partial f$ and let $r_u$ and $r_d$, respectively, be the first \svertex{} reached along this walk.
By Lemma~\ref{lemma:Clockwise_coloring}, we know that $r_u$ is blue whereas $r_d$ is red.

Recall that we want to construct a \smatching{} $M_f$ of $f$.
Let $V_L$ and $V_R$ be the sets of \svertices{} in $f$ that lie to the left and right, respectively, of the supporting line of $r_dr_u$, directed from $r_d$ to $r_u$.
Let $\pi\in\{L,R\}$.
Because $r_dr_u$ is a bichromatic edge, $V_\pi$ contains an even number of \svertices, half of them red and half of them blue.
As $V_\pi$ is a set in convex position, there exists a $BR$-matching on $V_\pi$.
Further, the convex hull of $V_\pi$ does not intersect $\ell$; see \figurename~\ref{fig:Switch-matchings}~(b).
Thus, for each $BR$-matching $M_\pi$ on $V_\pi$ no edge intersects $\ell$.

We obtain a \smatching{} $M_f$ of $f$ by taking the union of the edges of $M_L$ and $M_R$, and adding the edge $r_dr_u$, which is the only edge in $M_f$ intersecting $\ell$.
\end{proof}

\subsection{Putting things together}\label{section:Avoiding_ham-sandwich_cuts}

We proceed by showing how to obtain a $BR$-matching $M'$ on $P$ such that $M'$ and $G_M^{n}$ are compatible (and hence, $M'$ and $M$ are compatible) and $M'$ has fewer edges intersecting~$\ell$ than $M$ has. 
Recall that $C_{M,\ell}$ is the sequence of edges of $M$ that intersect $\ell$.
 
\begin{lemma}\label{lemma:Compatible matching with fewer crossings}
Let $M$ be a $BR$-matching on $P$ and let $\ell$ be a ham-sandwich cut of $P$.
There exists a $BR$-matching $M'$ compatible with $M$ such that $|C_{M', \ell}| \leq |C_{M, \ell}| - 2$.
\end{lemma}
\begin{proof}
For each face $f$ of $G_M^{n}$, consider a \smatching{} $M_f$ on the \svertices{} of $f$, such that the edges of $M_f$ have the minimum number of intersections with $\ell$.
Let $M'$ be the $BR$-matching which is the union of the edges of all these \smatchings{} $M_f$ for all faces $f$ of $G_M^{n}$.
Because every \smatching{} $M_f$ is contained in its respective face $f$, $M'$ and $G_M^{n}$ are compatible.
Moreover, since $M$ is contained in the boundary of $G_M^{n}$, $M'$ and $M$ are compatible.

Observe that edges of $M_f$ can intersect $\ell$ only if $f$ crosses $\ell$.
By Lemma~\ref{lemma:Few_faces_crossing_ell}, there are at most $k-1$ faces of $G_M^{n}$ that cross $\ell$, where $k = |C_{M, \ell}|$.
Furthermore, by Lemma~\ref{lemma:Switch-matchings}, each of these faces admits a \smatching{} having at most one edge intersecting $\ell$.
Therefore, $M'$ contains at most $k-1$ edges that intersect $\ell$.
However, every $BR$-matching must have an even number of edges that intersect $\ell$~\cite{bichromaticMatchings}.
Therefore, $M'$ contains at most $k-2$ edges that intersect $\ell$, proving our result.
\end{proof}

We are now ready to provide the proof of Lemma~\ref{lemma:CompatibleSequence} which is restated below.

\vspace{.15in}
\noindent \textbf{Lemma~\ref{lemma:CompatibleSequence}.}\emph{
Let $P= B\cup R$ be a bichromatic set of $2n$ points in the plane in general position such that $|B|=|R|=n$. 
For every $BR$-matching $M$ and every ham-sandwich cut $\ell$ of $P$, there exists a $BR$-matching $M^\ell$ such that $M$ and $M^\ell$ are $\lfloor n/2 \rfloor$-compatible and no edge of $M^\ell$ intersects~$\ell$.
}
\begin{proof}
Let $M_0 = M$ and $k=|C_{M, \ell}|$.
We know from Lemma~\ref{lemma:Compatible matching with fewer crossings} that for each $BR$-matching $M_i$ with $|C_{M_i, \ell}| > 0$ there exists a $BR$-matching $M_{i+1}$, such that $M_i$ and $M_{i+1}$ are compatible and $|C_{M_{i+1}, \ell}| \leq |C_{M_i, \ell}| - 2$.
Hence, there exists a transformation $M = M_0, \ldots, M_t = M^\ell$ of length~$t$ between $M$ and $M^\ell$, where $M^\ell$ contains no edge intersecting $\ell$.
As $|C_{M_{i+1}, \ell}| \leq |C_{M_i, \ell}| - 2$, for $0\leq i\leq t-1$, we conclude that $t \leq k/2 \leq n/2$, i.e., $M$ and $M^\ell$ are $\lfloor n/2 \rfloor$-compatible.
\end{proof}

{\small
\bibliographystyle{abbrv}
\bibliography{bichromaticMatchings}

\begin{thebibliography}{10}

\bibitem{Aichholzer2009617}
O.~Aichholzer, S.~Bereg, A.~Dumitrescu, A.~Garc\'ia, C.~Huemer, F.~Hurtado,
  M.~Kano, A.~M\'arquez, D.~Rappaport, S.~Smorodinsky, D.~Souvaine, J.~Urrutia,
  and D.~R. Wood.
\newblock Compatible geometric matchings.
\newblock {\em Computational Geometry}, 42:617 -- 626, 2009.

\bibitem{CompatibleMatchings2011}
O.~Aichholzer, A.~Garc{\'i}a, F.~Hurtado, and J.~Tejel.
\newblock Compatible matchings in geometric graphs.
\newblock In {\em Proceedings of the XIV Encuentros de Geometr{\'i}a
  Computacional}, pages 145--148, 2011.

\bibitem{CompatibleMatchingsForPSLG}
O.~Aichholzer, F.~Hurtado, and B.~Vogtenhuber.
\newblock Compatible matchings for bichromatic plane straight-line graphs.
\newblock In {\em Proceedings of the 28th European Workshop on Computational
  Geometry EuroCG '12}, pages 257--260, 2012.

\bibitem{bichromaticMatchings}
G.~Aloupis, L.~Barba, S.~Langerman, and D.~L. Souvaine.
\newblock Bichromatic compatible matchings.
\newblock In {\em Proceedings of the 29th annual Symposuim on Computational
  Geometry}, pages 267--276, 2013.

\bibitem{Borsuk33}
K.~Borsuk.
\newblock Drei {S}{\"a}tze {\"u}ber die $n$-dimensionale euklidische
  {S}ph{\"a}re.
\newblock {\em Fundamenta Mathematicae}, 20:177--190, 1933.

\bibitem{Goodman:1997:HDC:285869}
J.~E. Goodman and J.~O'Rourke, editors.
\newblock {\em Handbook of discrete and computational geometry}.
\newblock CRC Press, 1997.

\bibitem{HoffmannToth2013}
M.~Hoffmann and {\relax Cs}.~D. T\'oth.
\newblock Vertex-colored encompassing graphs.
\newblock {\em Graphs and Combinatorics}, 2013.
\newblock to appear (http://dx.doi.org/10.1007/s00373-013-1320-1).

\bibitem{Hurtado200814}
F.~Hurtado, M.~Kano, D.~Rappaport, and {\relax Cs}.~D. T\'oth.
\newblock Encompassing colored planar straight line graphs.
\newblock {\em Computational Geometry}, 39(1):14--23, 2008.

\bibitem{Ishaque2011}
M.~Ishaque, D.~L. Souvaine, and {\relax Cs}.~D. T\'oth.
\newblock Disjoint compatible geometric matchings.
\newblock {\em Discrete and Computational Geometry}, 49(1):89--131, 2013.

\bibitem{redbluesurvey}
A.~Kaneko and M.~Kano.
\newblock Discrete geometry on red and blue points in the plane---a survey.
\newblock In B.~Aronov, S.~Basu, J.~Pach, and M.~Sharir, editors, {\em Discrete
  and Computational Geometry, The Goodman-Pollack Festschrift}, volume~25 of
  {\em Algorithms and Combinatorics}, pages 551--570. Springer, 2003.

\bibitem{larson183}
C.~Larson.
\newblock {\em Problem Solving through Problems}.
\newblock Springer Publishing Company, Incorporated, 1983.

\bibitem{MatousekSteiger94}
C.-Y. Lo, J.~Matou{\v{s}}ek, and W.~Steiger.
\newblock Algorithms for ham-sandwich cuts.
\newblock {\em Discrete and Computational Geometry}, 11:433--452, 1994.

\bibitem{MatousekBorsukUlam}
J.~Matou{\v{s}}ek.
\newblock {\em Using the Borsuk-Ulam Theorem: Lectures on Topological Methods
  in Combinatorics and Geometry}.
\newblock Springer Publishing Company, Incorporated, 2007.

\bibitem{Razen2008}
A.~Razen.
\newblock A lower bound for the transformation of compatible perfect matchings.
\newblock In {\em Proceedings of the 24th European Workshop on Computational
  Geometry EuroCG '08}, pages 115--118, 2008.

\bibitem{Sharir2006}
M.~Sharir and E.~Welzl.
\newblock On the number of crossing-free matchings, cycles, and partitions.
\newblock {\em SIAM Journal on Computing}, 36(3):695--720, 2006.

\end{thebibliography}
}

\end{document}